\documentclass[draftclsnofoot,onecolumn]{IEEEtran}
\ifCLASSINFOpdf

\else

\fi

\usepackage{amssymb,amsfonts,amsmath,amstext,mathrsfs}
\usepackage{latexsym}
\usepackage{epsfig,psfrag,color,graphics,epsf}
\usepackage{enumerate,cite}
\usepackage{amsmath}
\usepackage{amssymb}
\usepackage{amsfonts}
\usepackage{amsthm}
\usepackage{graphicx}
\usepackage{float}
\usepackage{graphics}
\usepackage{graphicx}
\usepackage{euscript}
\usepackage{psfrag}
\usepackage{bbm}

\newtheorem{theorem}{Theorem}
\newtheorem{lemma}{Lemma}

\newtheorem{corollary}{Corollary}
\newtheorem{definition}{Definition}
\newtheorem{example}{Example}

\newcommand {\bq} {\mbox{\boldmath $q$}}
\newcommand {\bTheta} {\mbox{\boldmath $\Theta$}}

\newcommand {\bx} {\mbox{\boldmath $x$}}
\newcommand {\by} {\mbox{\boldmath $y$}}

\newcommand {\bP} {\mbox{\boldmath $P$}}

\newcommand {\bS} {\mbox{\boldmath $S$}}

\newcommand {\bW} {\mbox{\boldmath $W$}}
\newcommand {\bX} {\mbox{\boldmath $X$}}
\newcommand {\bY} {\mbox{\boldmath $Y$}}

\newcommand{\calC}{{\cal C}}
\newcommand{\calD}{{\cal D}}
\newcommand{\calE}{{\cal E}}

\newcommand{\calG}{{\cal G}}

\newcommand{\calL}{{\cal L}}

\newcommand{\calP}{{\cal P}}

\newcommand{\calS}{{\cal S}}

\newcommand{\calX}{{\cal X}}
\newcommand{\calY}{{\cal Y}}

\begin{document}
\title{Ratio List Decoding}

\author{
  Anelia Somekh-Baruch\thanks{A.\ Somekh-Baruch is with the Faculty of Engineering at Bar-Ilan University, Ramat-Gan, Israel.  Email: somekha@biu.ac.il. 
  This work was supported by the Israel Science Foundation (ISF) under grant 631/17. 
  Some of the results of this paper were presented at the IEEE International Symposium on Information Theory (ISIT) 2015. This paper was accepted for publication at the IEEE Transactions on Information Theory.}
}
\maketitle

\begin{abstract}

We extend the notion of list decoding to {\em ratio list decoding} which involves a list decoder whose list size is specified as a function of the number of messages $M_n$ and the block length $n$. We present necessary and sufficient conditions on $M_n$ for the existence of code sequences which enable reliable list decoding with respect to the desired list size $L(M_n,n)$. It is shown that the ratio-capacity, defined as the supremum of achievable normalized logarithms of the ratio $r(M_n,n)=M_n/L(M_n,n)$ is equal to the Shannon channel capacity $C$, for both stochastic and deterministic encoding. Allowing for random list size, we are able to deduce some properties of identification codes, where the decoder's output can be viewed as a list of messages corresponding to decision regions that include the channel output. 

We further address the regime of mismatched list decoding, in which the list constitutes of the codewords that accumulate the highest score values (jointly with the channel output) according to some given function. We study the case of deterministic encoding and mismatched ratio list decoding. We establish similar necessary and sufficient conditions for the existence of code sequences which enable reliable mismatched list decoding with respect to the desired list size $L(M_n,n)$, and we show that the ratio-capacity with mismatched decoding is equal to the mismatch capacity. Focusing on the case of an exponential list size $L_n=e^{n \Theta}$, its comparison with ordinary mismatched decoding shows that the increase in capacity is by $\Theta$ bits per channel use for all channels and decoding metrics. Several properties of the average error probability in the setup of mismatched list decoding with deterministic list size are provided.  

\end{abstract}

\newpage

\section{Introduction}
Unlike ordinary decoding where the decoder has to select a single message as its estimate, a list decoder outputs a list of messages, among which the transmitted one is expected to be found. Therefore, an error occurs if the actual transmitted message is not in the list. 
List decoding was introduced by Elias \cite{Elias1957} and by Wozencraft \cite{Wozencraft58}, and has been studied extensively
for linear codes and other specific code structures
(see \cite{sudan2000list,ChunlongBaiMielczarekKrzymienFair2007,Guruswami2001,Guruswami_FnT_2007} and references therein), and also from the information theoretic point of view (see
\cite{ShannonGallagerBerlekamp1967,Forney1968Erasure,HofSasonShamai_IT2010,Merhav_List_Decoding_IT_2014} and references therein). The notion of a decoder that outputs a list of possible messages arises naturally in many channel coding settings, either for applications that do not require full decoding of the transmitted message, or when full decoding is impossible. 
For example, the relay channel can be regarded as a case in which list decoding is used, even though the ultimate goal of the receiver is to obtain a single message and not a list of messages \cite{CoverElGamalRelay1979}. 

In certain applications, it is possible to pick the correct codeword from the list with the help of side information or the semantic context. In some cases, such as concatenated code constructions, the entire list is more advantageous than just having the most probable codeword \cite{Guruswami_FnT_2007,TalVardy2015}.

In this paper we extend list decoding  
to a notion that we call ratio list decoding. 
Whereas classical list decoding has to do with a list size which is predetermined as a function of the block length $n$, in ratio list decoding, the list size is a function of both $n$ and the size of the message set $M_n$.
As in classical list decoding, the requirement from the decoder is that with probability converging to one, the transmitted codeword must belong to the declared list whose size is a function of $(n,M_n)$.

For simplicity of presentation, we consider a desired list size of $L(M_n,n)=\frac{M_n}{r(M_n,n)}$ where $r(M_n,n)$ is referred to as the ratio function
(ratio of codebook size to list size). 
We allow the actual size of the list to be a random variable with a vanishingly small probability  to exceed $L(M_n,n)$. 
Note that the following three cases fall within the scope of our model: \begin{itemize}
\item The case of $r(M_n,n)=M_n$, ($L(M_n,n)=1$) which corresponds to classical channel coding. 
\item The cases in which $L(M_n,n)$ is a function of the block length, such as $L(M_n,n)=\exp(n\Theta)$; i.e., $r(M_n,n)=M_n\cdot \exp(-n\Theta)$ have been studied extensively in the literature. 
\item The special case of $r(M_n,n)=1$, i.e., $L(M_n,n)=M_n$, where clearly the codebook size can be infinite, since the list of messages $\{1,...,M_n\}$ is exhaustive. 
\end{itemize}

\vspace{0.1cm}

The new setup of ratio list decoding introduces a generalized theoretical perspective to list decoding, by adding the dimension of the proportion of the list size compared to the number of messages. 
It is motivated by applications such as concatenated code constructions, in which the proportion of the list relatively to the entire message set is of main interest.

We show that under stochastic as well as deterministic encoding, the supremum of the achievable normalized logarithm of the codebook to the list size ratios is equal to the Shannon channel capacity $C$. 
Furthermore, we show that if the number of messages as a function of the block length $M_n$ is such that $\limsup_{n\rightarrow\infty}\frac{1}{n}\log r(M_n,n)>C$, reliable list decoding cannot occur, and if $0< \limsup_{n\rightarrow\infty}\frac{1}{n}\log r(M_n,n)<C$ then there exists a sequence of codes having $M_n$ messages for blocklength $n$ which enables reliable list decoding w.r.t.\ (with respect to) the ratio function $r(M_n,n)$. 
We first prove these results for information-stable channels, and subsequently we show that they continue to hold for general channels in the Verd\'{u}-Han sense \cite{VerduHan1994}.

As a corollary of our results we deduce some properties of identification codes. 
Identification codes \cite{AhlswedeDueck1989} are about a decoder that needs to answer reliably to the $M_n$ binary hypothesis testing questions ``was message $i$ transmitted?'' for $i=1,...,M_n$. 
Although this is not a list decoding setup, the decoder's output can be viewed as a list of messages whose decision regions include the channel output $Y^n$. 
This interpretation enables us to derive a non-vanishing lower bound on the probability that the list size of the identification decoder with $M_n>\exp\left\{e^{n(C-\epsilon)}\right\}$ will exceed $\exp\left\{e^{n(C-\epsilon)}\right\}/e^{n(C-\epsilon-\delta)}$ for $\delta>\epsilon>0$.

The remainder of this paper focuses on the important special case in which the structure of the decoder is predetermined, such as for instance the Hamming distance, and cannot be optimized with respect to the actual channel over which transmission occurs. We refer to this setup as mismatched list decoding. 
This setup extends classical mismatched decoding to the framework of list decoding. In classical mismatched decoding, there is a real-valued function $q_n$, usually referred to as a ``metric'', which maps each pair of channel input and output sequences $(x^n,y^n)$ to a real number $q_n(x^n,y^n)$. The decoder chooses the message $\hat{m}=\mbox{argmax}_{i\in \{1,...,M_n\}}q_n(x^n(i),y^n)$ as its output. It is assumed that in selecting its codebook, the encoder is aware of the structure of the decoder. 
Mismatched decoding for the Discrete Memory Channel (DMC) with an additive metric $q_n(x^n,y^n)=\sum_{i=1}^nq(x_i,y_i)$ was studied by Csisz\'{a}r and K{\"o}rner \cite{CsiszarKorner81graph} and by Hui \cite{Hui83} who presented a formula for the rate achievable by random coding. It turns out that higher rates can be achieved by more complex random coding methods \cite{CsiszarNarayan95,Lapidoth96} such as superposition coding \cite{ScarlettMartinezGuilleniFabregasISIT_2013,ScarlettMartinezGuilleniFabregas_mismatch_2014_IT,
ScarlettMartinezGuilleniFabregas2012AllertonSU,
ScarlettPengMerhavMartinezGuilleniFabregas_mismatch_2014_IT,SomekhBaruch_mismatchachievableIT2014,
SomekhBaruchISIT_2013}. Achievable error exponents were studied in \cite{CsiszarKorner81graph,ScarlettPengMerhavMartinezGuilleniFabregas_mismatch_2014_IT,SomekhBaruch_mismatchachievableIT2014}. Nevertheless, a single letter formula for the mismatch capacity of the DMC has not been established. Multiletter upper bounds are obtained in \cite{SomekhBaruchMismatchedDMCIT2018_accepted} and a general multi-letter expression for the mismatch capacity was derived in \cite{SomekhBaruch_general_formula_IT2015}. For other related works, see  
\cite{MerhavKaplanLapidothShamai94,Lapidoth96b,
GantiLapidothTelatar2000,ShamaiSason2002} and references therein.

In our mismatched list decoding, 
the decoder's list is composed of the codewords which accumulate the highest metrics values jointly with the channel output, and in this setup  
we focus on deterministic encoding. The main reason for focusing on deterministic encoding, apart of the simplicity of the analysis, is that it is not clear how one should define a fixed decoding rule when there are multiple options for transmitted signal $x^n$ which occur with positive probability given a particular message $m$. There are several possibilities to approach this issue such as considering the decoding rules:
 $\hat{m}=\mbox{argmax}_m \max_{x^n:\; \Pr(x^n|S=m)>0}q_n(x^n,y^n)$ or 
$\hat{m}=\mbox{argmax}_m \sum_{x^n} \Pr(x^n|S=m)\cdot q_n(x^n,y^n)$ (where $S$ stands for the random message) but the analysis thereof becomes more involved. 

In the mismatched case too, we  
establish necessary and sufficient conditions 
for the existence of code sequences which enable reliable mismatched list decoding with respect to the desired list size $L(M,n)$. We show that in this mismatched setup, the supremum of achievable logarithm of normalized ratios $\frac{1}{n}\log r(M_n,n)$ is equal to the mismatch capacity.

Further, we specialize the results to the mismatched case of a list of size $L_n=e^{n\Theta_n}$ where $\Theta_n\in(0,c]$, in which case we derive a general multi-letter formula for the capacity. 
This is an extension of our previous results in \cite{SomekhBaruch_general_formula_IT2015}, where a general multi-letter formula was established for the mismatch capacity of a general channel, defined as a sequence of conditional distributions with a general decoding metric sequence.
It is shown that the increase in capacity for the $e^{n\Theta}$ list size case (compared to ordinary mismatched decoding) is $\Theta$ bits per channel use. 

An expression for the average error probability in list decoding with {\it a constant} list size $L_n=e^{n\Theta_n}$ (that is, equal list size for all $y^n$) and rate $R$, denoted ${\cal E}_{q_n}^{(n)}(R,\Theta_n)$, where $\Theta_n\leq R$, is established. 
We further present a random coding lower bound on ${\cal E}_{q_n}^{(n)}(R,\Theta_n)$ which is based on the analysis of \cite{Merhav_List_Decoding_IT_2014}.
Finally, we derive an inequality that can be regarded as an extension (to the case of mismatched list decoding) of the inequality resulting from Fano's inequality for matched channel coding for the DMC; i.e., 
$
P_e^{(n)}\geq 1-\frac{C}{R}-\frac{1}{nR}$, where $C$ is the channel capacity, $R$ is the code rate and $P_e^{(n)}$ is the average probability of error obtained by a code of rate $R$. 
In the case of the erasures-only decoding metric, this yields a lower bound on the average error probability above the capacity.

This paper is organized as follows. 
Section \ref{sc: Problem Formulation and Definitions} is devoted to definitions and the problem formulation. 
Section \ref{sc: Converse and Direct Results for Ratio List Decoding} presents converse and direct results for ratio list decoding. 
In 
Section \ref{sc: Mismatched List Decoding} we extend the results to the mismatched case, 
and describe 
properties of the average error probability in mismatched list decoding. 
Section \ref{sc: discussion} presents the discussion and concluding remarks. 

\section{Problem Formulation and Definitions}\label{sc: Problem Formulation and Definitions}

We consider a point-to-point communication channel with input alphabet $\calX$ and output alphabet $\calY$. 
We adopt the following definition of \cite{VerduHan1994} for a general channel. 
A channel $\bW=\{W^{(n)}\}_{n=1}^{\infty}$ is an arbitrary sequence of increasing dimension where $W^{(n)}$ is a conditional output
distribution from $\calX^n$ to $\calY^n$, and where $\calX$ and $\calY$ are the
input and output alphabets, respectively. With a little abuse of terminology we refer to $\bW$ as well as to $W^{(n)}$ as channels, where the exact meaning will be clear from the context. 
 
An encoder observes a random message $S$, which is distributed uniformly over $\{1,...,M_n\}$, and produces a channel input signal $X^n$ of length $n$ as a function of $S$.\footnote{In fact, it should be understood that $S$ depends on $n$, and should be denoted as $S_n$, but for simplicity of presentation we omit the dependence of $S_n$ on $n$ from the notation whenever possible.} The encoder is not constrained to using deterministic functions and therefore can be viewed as a collection of distributions $\{P_n(\cdot|m)\}$ from $\{1,...,M_n\}$ to $\calX^n$. 
The signal $X^n$ is fed into the channel $W^{(n)}$, which is a conditional distribution from $\calX^n$ to $\calY^n$, and the resulting channel output is denoted by $Y^n$.

A list decoder is defined by a collection of not necessarily disjoint decision regions
\begin{flalign}
\calD_m\subseteq \calY^n, \; m\in\{1,...,M_n\}.
\end{flalign}
Another representation of the decoder can be the list of decision regions which contain the channel output. In other words, we let $L_n(y^n)$ be the list at the output of the decoder, 
that is, 
\begin{flalign}
L_n(y^n)=& \left\{ m\in\{1,...,M_n\}:\; y^n\in \calD_m\right\}. 
\end{flalign}
A list decoder maps every $y^n\in\calY^n$ into a subset $L_n(y^n)\subseteq \{1,...,M_n\}$.
Note that we allow the list size to depend on $y^n$, and that
\begin{flalign}
y^n\in\calD_S \Leftrightarrow S\notin L_n(y^n).
\end{flalign}

We next state a number of definitions leading to the ratio-capacity. Let $\mathbb{N}$ stand for the set of all positive integers.
\begin{definition}
 We say that $r(M_n,n)$ is a ratio function if for all $(n,M_n)\in\mathbb{N}\times \mathbb{N}$ one has $1\leq r(M_n,n)\leq M_n$. 
 
\end{definition}

\begin{definition}\label{df: a code parameters}
An $(n,M_n,\epsilon,\zeta)$-code $\left\{\{P(\cdot|m)\}_{m\in \{1,...,M_n\}},\; \{L_n(y^n)\}_{y^n\in\calY^n} \right\}$ for the channel w.r.t.\ ratio function $r(M_n,n)$ is one for which
\begin{flalign}\label{eq: two probabilities}
&\Pr\left( S\notin L_n(Y^n)\right)= \epsilon, \mbox{ and}\nonumber\\
& \Pr\left(|L_n(Y^n)|>\frac{M_n}{r(M_n,n)}\right)= \zeta.
\end{flalign}
In words, $\epsilon$ is the probability that the transmitted message is not in the output list of the decoder, and $\zeta$ is the probability that the output list size is larger than permitted by the ratio function. The quantities $\epsilon$ and $\zeta$ will be referred to as error of the first and second kind, respectively. 
\end{definition}

\begin{definition}\label{df: achievable rho}
We say that a sequence $\{M_i\}_{i=1}^{\infty}$, $M_i\in\mathbb{N}$ is feasible for the channel w.r.t.\ ratio sequence $r(M_n,n)$ if
there exists a sequence of $(n,M_n,\epsilon_n,\zeta_n)$-codes
having vanishing probabilities of error of the first and second kinds, that is,  
$\lim_{n\rightarrow\infty} \epsilon_n=0$, and $ \lim_{n\rightarrow\infty} \zeta_n=0$.
\end{definition}
\begin{definition}
We say that $\rho$ is an achievable normalized log-ratio for the channel if there exist a ratio function $r(M_n,n)$ and a corresponding feasible sequence $M_1,M_2,....$, $M_i\in\mathbb{N}$ such that $\limsup_{n\rightarrow\infty}\frac{1}{n}\log r(M_n,n)= \rho$.
\end{definition}
\begin{definition}\label{df: maximal achievable rho}
The {\it ratio-capacity} of the channel $\bW$ is the supremum of achievable normalized log-ratios, and will be denoted $\rho_{sup}(\bW)$. 
\end{definition}

\section{Converse and Direct Results for Ratio List Decoding}\label{sc: Converse and Direct Results for Ratio List Decoding}

Our first result is a coding theorem for channels $\bW$ whose capacity is given by the formula
\begin{flalign}\label{eq: inf stab cond}
C= \limsup_{n\rightarrow\infty} \max_{P(X^n)} \frac{1}{n}I(X^n;Y^n).
\end{flalign}
Note that information-stable channels satisfy this condition (see \cite{VerduHan1994}). This enables us to use a Fano-type proof for the converse part, and to present some insights that are applicable to identification codes.
In Section \ref{sc: An Extension to General Channels} we extend the coding theorem to general channels in the spirit of \cite{VerduHan1994}.

\begin{theorem}\label{th: main theorem}
Let $\{M_i\}_{i=1}^{\infty}$, $M_i\in\mathbb{N}$, be a given sequence and let $r(M_n,n)$ be a ratio function. 
Let 
\begin{flalign}\label{eq: rho r dfn}
\rho_r=\limsup_{n\rightarrow\infty}\frac{1}{n}\log r(M_n,n). 
\end{flalign}
(i) Converse part: If the sequence $\{M_i\}_{i=1}^{\infty}$ is feasible for the channel w.r.t.\ the
 ratio function $r(M_n,n)$ and if the channel capacity $C$ satisfies (\ref{eq: inf stab cond})
then 
\begin{flalign}\label{eq: converse part ofTh1}
\rho_r\leq  C.
\end{flalign}
(ii) Direct part: If the sequence $\{M_i\}_{i=1}^{\infty}$ satisfies 
\begin{flalign}\label{eq: direct part ofTh1}
0<\rho_r<C,
\end{flalign}
where $C$ is the channel capacity, it is feasible for the channel w.r.t.\ the
 ratio function $r(M_n,n)$. 
\end{theorem}
The theorem straightforwardly implies the following theorem. 
\begin{corollary}\label{th: rho max theorem 1}
If the capacity of the channel $\bW$ satisfies (\ref{eq: inf stab cond}), then 
\begin{flalign}
 \rho_{sup}(\bW)=C. 
 \end{flalign}
 \end{corollary}
 Note that in fact $C$ should be denoted $C(\bW)$, but this is omitted for notational convenience. 
The converse part of Corollary \ref{th: rho max theorem 1} follows from (\ref{eq: converse part ofTh1}), and the direct part follows by choosing $r(M_n,n)=M_n$ and by using the direct part of the ordinary channel coding theorem. 

Before we prove Theorem \ref{th: main theorem} we present the following example which demonstrates the application of the theorem. 

\begin{example}
It is required to construct a sequence of codes such that reliable list decoding is possible with list size $L(M_n,n)=M_n^{1-1/n} $. Theorem \ref{th: main theorem}  implies that the maximal number of messages $M_n$ as a function of the blocklength $n$ is essentially $M_n\approx e^{n^2C}$; that is,  for all $\epsilon>0$ 
there exists a sequence of $(n,e^{n^2(C-\epsilon)},\epsilon_n,\zeta_n)$ codes such that $\lim_{n\rightarrow\infty} \epsilon_n=0$, and $ \lim_{n\rightarrow\infty} \zeta_n=0$, and no such sequence exists if $M_n>e^{n^2(C+\epsilon)}$.
To see this, note that the (monotonically increasing) function $M_n=e^{n^2C}$ is the solution of the equation
$
\frac{1}{n}\log r(M_n,n)=\frac{1}{n}\log \left(M_n/M_n^{1-1/n}\right)=C$.
\end{example}

We begin with the converse part. 

\subsection{A proof of the Converse Part of Theorem \ref{th: main theorem} for Information Stable Channels and Implications}
The proof of the following Fano-type inequality for list decoding appears in \cite[Appendix 3.E]{RaginskySason_FnT}.\footnote{The inequality in \cite{RaginskySason_FnT} involves a discrete random variable $Y$  but it is easy to realize that in fact $Y$ need not necessarily be discrete for the proof to hold.} 
\begin{lemma}\label{lm: Fanolike lemma} 
Let $X$ be a discrete random variable over alphabet $\calX$. Let $Y$ be another random variable, and let $\calL(Y)\subseteq \calX$ be a mapping from $\calY$ to the set of all subsets of $\calX$.\footnote{One can think of $\calL(Y)$ as a set of estimators of $X$ which lie in $\calX$.} Denote 
\begin{flalign}
P_e=& \Pr\left(X\notin \calL(Y) \right)
\end{flalign}
 then
\begin{flalign}\label{eq: FnoFano}
&H(X|Y) \leq  h_2(P_e) +P_e\cdot \log(|\calX|-1)+(1-P_e)\mathbb{E}(\log|\calL(Y)|),
\end{flalign}
where $h_2(\cdot)$ is the binary entropy function; i.e., $h_2(t)=-t\log(t)-(1-t)\log(1-t)$. 
\end{lemma}
The generalization of Lemma \ref{lm: Fanolike lemma} for fixed-size list decoding (i.e., when $ | \mathcal{L}(Y) | $ is independent of observation $Y$), which provides an upper bound on the Arimoto-R\'{e}nyi conditional entropy of an arbitrary positive order $\alpha$ as a function of the list decoding error probability, appears in \cite[Section 4]{SasonVerdu2017ArimotoRenyi_IT}. This result generalizes previously reported bounds in \cite[Section 5]{AhlswedeGacsKorner1976} and \cite[Section 1]{KimSutivongCover2008}, providing an upper bound on the conventional conditional entropy of Shannon (the latter being a special case of the Arimoto-R\'{e}nyi conditional entropy for $\alpha=1$) in the setting of fixed-size list decoding.

Our next result uses Lemma \ref{lm: Fanolike lemma} to upper bound $\frac{1}{n}\log(r(M_n,n))$.
\begin{theorem}\label{th: list theorem}
Let a channel, a code $\{P(\cdot |m)_{m\in\{1,...,M_n\}} ,\; L_n(y^n)_{y^n\in\calY^n}\}$, and a ratio function $r(M_n,n)$ be given. 
The following inequality holds
\begin{flalign}\label{eq: yom bassa}
\frac{1}{n}\log(r(M_n,n)) \leq \frac{1}{n}  \frac{  I(X^n;Y^n)+1}{(1-\epsilon_n) (1-\zeta_n)} 
\end{flalign}
where $X^n$ and $Y^n$ are the channel input and output vectors, respectively, and
\begin{flalign}\label{eq: epsilon zeta dfns}
\epsilon_n=& \Pr(S\notin L_n(Y^n)),  \nonumber\\
\zeta_n\triangleq & \Pr\left(|L_n(Y^n)|>\frac{M_n}{r(M_n,n)}\right).
\end{flalign}
\end{theorem}
\begin{proof}
Since $S-X^n-Y^n$ is a Markov chain we have
\begin{flalign}\label{eq: C bound}
I(S;Y^n)
\leq & I(X^n;Y^n).
\end{flalign}
From Lemma \ref{lm: Fanolike lemma} applied to $S,Y^n,L_n(Y^n)$ in the role of $X,Y,\calL(Y)$, respectively, we obtain
\begin{flalign}\label{eq: applying the lemma}
H(S|L_n(Y^n))\leq 1+\epsilon_n\cdot \log M_n +(1-\epsilon_n)\mathbb{E}\left(\log\left(|L_n(Y^n)|\right)\right).
\end{flalign}
Therefore,
\begin{flalign}
&\log(M_n)\nonumber\\
&= H(S)\nonumber\\
&= I(S;Y^n) +H(S|Y^n)\nonumber\\
&\leq I(X^n;Y^n)+H(S|Y^n)\label{eq: I let H le}\\
&\leq  I(X^n;Y^n)+1+\epsilon_n\cdot \log M_n +(1-\epsilon_n)\mathbb{E}\left(\log\left(|L_n(Y^n)|\right)\right),\label{eq: result1}
\end{flalign}
where (\ref{eq: I let H le}) follows from (\ref{eq: C bound}), and (\ref{eq: result1}) follows from (\ref{eq: applying the lemma}).
 
Now, since $|L_n(Y^n)|\leq M_n$, by definition of $\zeta_n$ we obtain
 \begin{flalign}\label{eq: result2}
&\mathbb{E}\left(\log|L_n(Y^n)|\right)
\leq  \zeta_n \cdot \log(M_n) + (1-\zeta_n)\cdot \log\left(\frac{M_n}{r(M_n,n)}\right). 
\end{flalign}
Consequently, from (\ref{eq: result1})-(\ref{eq: result2}) we get 
\begin{flalign}\label{eq: result3}
&\log(M_n)\leq   I(X^n;Y^n)+1+\epsilon_n\cdot \log M_n 
+(1-\epsilon_n)\left[\zeta_n \cdot \log(M_n) + (1-\zeta_n)\cdot \log\left(\frac{M_n}{r(M_n,n)}\right) \right].
\end{flalign}
Hence,
\begin{flalign}\label{eq: result4}
&\left(1-\epsilon_n \right)(1-\zeta_n)\cdot \log(M_n) \leq   I(X^n;Y^n)+1 
+(1-\epsilon_n) (1-\zeta_n)\cdot \log\left(\frac{M_n}{r(M_n,n)}\right).
\end{flalign}
That is, 
\begin{flalign}\label{eq: result5}
\frac{1}{n}\log(r(M_n,n)) \leq   \frac{1}{n}\frac{I(X^n;Y^n)+1}{(1-\epsilon_n) (1-\zeta_n)} ,
\end{flalign}
and this concludes the proof of Theorem \ref{th: list theorem}. 
\end{proof}
This straightforwardly leads to the converse part (i) of Theorem \ref{th: main theorem} which holds for channels that satisfy (\ref{eq: inf stab cond}); that is, 
if the sequence $\{M_i\}_{i=1}^{\infty}$ is feasible for the channel w.r.t.
 ratio function $r(M_n,n)$ then 
$
\rho_r\leq  C$.

A few conclusions can be drawn from the converse part of Theorem \ref{th: main theorem} concerning special cases of list decoding setups. 
\begin{itemize}
\item The converse part of Theorem \ref{th: main theorem} highlights the tradeoff between the probability of successful list decoding, ($1-\epsilon_n$), and the probability of not exceeding the desired ratio function, ($1-\zeta_n$). In particular, rearranging (\ref{eq: yom bassa}) we obtain
\begin{flalign}\label{eq: yom bassa beemet}
(1-\epsilon_n) (1-\zeta_n) \leq  \frac{  I(X^n;Y^n)+1/n}{\log(r(M_n,n))} ,
\end{flalign}
which gives a lower bound on $\epsilon_n$ for fixed  $\zeta_n$ and vice versa. 
\item Consider for example the case $|L(Y^n)|=e^{n\Theta}$; that is, $r(M_n,n)=\frac{M_n}{e^{n\Theta}}$. In this case, defining $R=\frac{1}{n}\log M_n$ we obtain from (\ref{eq: converse part ofTh1}) 
\begin{flalign}\label{eq: exponential size of list}
R\leq   C+\Theta,
\end{flalign}
which  for the special case of discrete memoryless channels is a known result. 

\item It is easy to see that if $r(M_n,n)=1$; that is, the list size is equal to $M_n$ and $\rho_r=0$, then the converse part $
\rho_r\leq  C$ poses no restriction on the number of messages, which is not surprising since a list that contains the entire set of all possible messages always includes the transmitted one, and therefore infinitely many messages can be transmitted.
\end{itemize}

The following corollary gives a lower bound on the probability that the list size exceeds $e^{-n\delta}L(M_n,n)$ if the normalized logarithm of the ratio $r(M_n,n)$ exceeds $C-\epsilon$ for $\delta>\epsilon>0$.

\begin{corollary}\label{cr: corolalldjfbv}
If $C=\limsup_{n\rightarrow\infty} \max_{P(X^n)} \frac{1}{n}I(X^n;Y^n)$, and a sequence of $(n,M_n,\epsilon_n,\zeta_n)$-codes is given such that 
$\limsup_{n\rightarrow\infty}\frac{1}{n}\log r(M_n,n)\geq C-\epsilon$, and $\lim_{n\rightarrow\infty}\epsilon_n=0$, 
then for all $\delta>0$
\begin{flalign}\label{eq: corollary result 4}
\limsup_{n\rightarrow\infty} \Pr\left(|L_n(Y^n)|>\frac{M_n}{e^{n\delta}r(M_n,n)}\right) \geq 1-\frac{C}{C+\delta-\epsilon}   .
\end{flalign}
\end{corollary}
Before proving Corollary \ref{cr: corolalldjfbv}, we mention its implication to the identification problem. To this end, we briefly describe the information-theoretic identification problem \cite{AhlswedeDueck1989}.  
A randomized identification code of parameters $(n,M_n)$ for the DMC $W$ from $\calX$ to $\calY$ is defined  by a family of $M_n$ conditional distributions $P(\cdot|i)\in\calP(\calX^n)$ and $M_n$ decision regions $\calD_i\subseteq \calY^n$ where $i=1,...,M_n$. The decision regions are not necessarily disjoint and the question is how large $M_n$ can be such that 
\begin{flalign}
&\sum_{x^n\in\calX^n } P(x^n|i)W^n(\calD_i^c|x^n)\leq \lambda_1, \mbox{ and}\nonumber\\  
&\sum_{x^n\in\calX^n } P(x^n|i)W^n(\calD_j|x^n)\leq \lambda_2,\;  \forall i\neq j\in \{1,...,M_n\}.
\end{flalign} 
It turns out that unlike the classical channel decoding problem, 
in order to guarantee arbitrarily small $\lambda_1$ and $\lambda_2$ for the DMC, $M_n$ can grow double-exponentially fast with $n$ with normalized iterated logarithm\footnote{The iterated logarithm of $M_n$ stands for $\log(\log(M_n))$. } that is equal to Shannon's channel capacity $C$. 

Note that 
by choosing $\delta>\epsilon>0$ and $r(M_n,n)=\log(M_n)$, Corollary \ref{cr: corolalldjfbv} implies that a necessary condition for obtaining high values of a normalized iterated logarithm of the number of messages of the identification code is a non-vanishing 
probability that the list will include essentially $\approx\frac{M_n}{\log(M_n)}$ messages. We state this special case as a separate corollary.
\begin{corollary}\label{cr: ID corollayr}
If $C=\limsup_{n\rightarrow\infty} \max_{P(X^n)} \frac{1}{n}I(X^n;Y^n)$ and a sequence of $(M_n,n)$-identification codes is given such that
\begin{flalign}
 \frac{1}{n}\log\log(M_n)\geq C-\epsilon
\end{flalign}
then 
\begin{flalign}\label{eq: corollary result 4}
\limsup_{n\rightarrow\infty} \Pr\left(|L_n(Y^n)|>\frac{\exp\{e^{n(C-\epsilon}\}}{e^{n(C-\epsilon+\delta)}}\right) \geq 1-\frac{C}{C+\delta-\epsilon}   .
\end{flalign}
\end{corollary}
This result implies that while a reliable identification code whose ID-rate is high does narrow down significantly the number of hypothesized messages with high probability, the list size of false positive message identifications exceeds $M_n/(\log(M_n)e^{n\delta})$ (and grows double-exponentially with $n$).

We next prove Corollary \ref{cr: corolalldjfbv}. 
\begin{proof}
Denote
\begin{flalign}
\eta_n\triangleq & \Pr\left(|L_n(Y^n)|>\frac{M_n}{\exp(n\delta)r(M_n,n)}\right), 
\end{flalign}
and repeat the proof of Theorem \ref{th: list theorem} (steps (\ref{eq: C bound})-(\ref{eq: result5})) with $r(M_n,n)$ replaced by $\exp(n\delta)\cdot r(M_n,n)$ and $\zeta_n$ replaced by $\eta_n$. This gives
\begin{flalign}\label{eq: corollary result}
\frac{1}{n}\log r(M_n,n) \leq   \frac{1}{n}\frac{I(X^n;Y^n)+1}{(1-\epsilon_n) (1-\eta_n)}-\delta .
\end{flalign}
Taking the limit as $n$ tends to infinity, and noting that in fact $\epsilon_n=\Pr(S\notin \calD_\calS)$ which is required to vanish as $n$ tends to infinity for reliable identification, we obtain
\begin{flalign}\label{eq: corollary result 2}
C-\epsilon \leq   \frac{C}{1-\limsup_{n\rightarrow\infty}\eta_n }-\delta,
\end{flalign}
which implies that 
\begin{flalign}\label{eq: corollary result 3}
 \limsup_{n\rightarrow\infty}\eta_n\geq 1-\frac{C}{C+\delta-\epsilon}   .
\end{flalign}

\end{proof}

\subsection{Proof of the Direct Part of Theorem \ref{th: main theorem}}

\begin{proof}
The proof of the direct part of Theorem \ref{th: main theorem} is quite straightforward. Assume a sequence $\{M_i\}_{i=1}^{\infty}$ satisfies (\ref{eq: direct part ofTh1}) and let 
\begin{flalign}\label{eq:  kfbvbdjsb}
&\Delta=  C-\limsup_{n\rightarrow\infty} \frac{1}{n}\log r(M_n,n)=C-\rho_r.
\end{flalign}
 Note that $\Delta\in (0,C)$. 
Let $\epsilon\in(0,\Delta/2]$ be an arbitrarily small constant and consider a sequence of classical channel encoders 
$\varphi_n:\; \{1,...,2^{n(C-\Delta+\epsilon)}\}\rightarrow \calX^n$, or codebooks $
\{x^n(m)\}, m=1,...,2^{n(C-\Delta+\epsilon)}$, having an average probability of error $\epsilon_n=\Pr(\hat{S}(Y^n)\neq S )$ such that $\lim_{n\rightarrow\infty}\epsilon_n=0$ (where $\hat{S}(Y^n)$ is the maximum likelihood decoder's output). 

Now, 
for blocklength $n$, 
construct a larger codebook which contains $\frac{M_n}{r(M_n,n)}$ identical copies of each of the codewords $\{x^n(m)\}$. 
The larger codebook corresponds to $M'_n=\frac{M_n}{r(M_n,n)}\cdot2^{n(C-\Delta+\epsilon)}$ messages.
From (\ref{eq:  kfbvbdjsb}) we know that for all sufficiently large $n$, $r(M_n,n)\leq 2^{n(C-\Delta+\epsilon/2)}$ and therefore, for all sufficiently large $n$,  
$
M_n'\geq M_n$. 
Since the transmitted codeword can be decoded successfully with probability $1-\epsilon_n$, the list decoder can output a list containing all the messages that are mapped to that codeword. 

To obtain a codebook of $M_n$ codewords, we simply use a subset of the $M_n'$ codewords (of size $M_n$) which yield an average probability of error at least as low as $\epsilon_n$. 
\end{proof}

\subsection{An Extension to General Channels}\label{sc: An Extension to General Channels}
While the direct result of Theorem \ref{th: main theorem} (ii) holds for every general channel, the converse result (i) is stated only for channels that satisfy $C=\limsup_{n\rightarrow\infty} \max_{P(X^n)} \frac{1}{n}I(X^n;Y^n)$. 
We next present a converse that generalizes Theorem \ref{th: main theorem} (i) to general channels which do not necessarily satisfy this condition. 

\begin{theorem}\label{th: main theorem 2}

(a)
If the sequence $\{M_i\}_{i=1}^{\infty}$ is feasible for the channel w.r.t.\ the ratio function $r(M_n,n)$, then 
\begin{flalign}\label{eq: converse part ofTh1 b}
\rho_r\leq  C,
\end{flalign}
where $\rho_r$ is defined in (\ref{eq: rho r dfn}). 

(b)
If $\left\{\{P(\cdot|m)\}_{m\in \{1,...,M_n\}},\; \{L_n(y^n)\}_{y^n\in\calY^n}\}\right\}$, n=1,2,..., is a sequence of $(n,M_n,\epsilon_n,\zeta_n)$-codes 
such that $\lim_{n\rightarrow\infty }\epsilon_n=0$ and
\begin{flalign}\label{eq: assumption general channels}
&\limsup_{n\rightarrow\infty} \frac{\mathbb{E}(|L_n(Y^n)|)}{L(M_n,n)}<\infty,
\end{flalign}
where $L(M_n,n)=M_n/r(M_n,n)$, then, 
\begin{flalign}\label{eq: converse part ofTh1 b c}
\rho_r\leq  C.
\end{flalign}
\end{theorem}
Consequently the following corollary follows in the same manner as that of Corollary \ref{th: rho max theorem 1}. 
\begin{corollary}\label{th: rho max theorem 2}
For every channel $\bW$
\begin{flalign}
 \rho_{sup}(\bW)=C. 
 \end{flalign}
\end{corollary}
Before proving 
Theorem \ref{th: main theorem 2}, it is noted that the converse part of Theorem \ref{th: main theorem} follows from Theorem \ref{th: main theorem 2} (a). Nevertheless, we included the proof of the converse part of Theorem \ref{th: main theorem} since it highlights the tradeoff between the probability of successful list decoding, and the probability of not exceeding the desired ratio function. Moreover, Corollaries \ref{cr: corolalldjfbv} and \ref{cr: ID corollayr} result from the proof of the converse part of Theorem \ref{th: main theorem}. Finally, the proof of Theorem \ref{th: main theorem 2} relies on the information-spectrum method, and the proof of the converse part of Theorem \ref{th: main theorem} relies on a Fano-like inequality. The latter gives insight on how the classical converse theorem proof is extended to the ratio list decoding case.

We next prove Theorem \ref{th: main theorem 2}. 
\begin{proof}
By assumption in part (a), since the sequence $\{M_i\}_{i=1}^{\infty}$ is feasible w.r.t.\ the ratio function $r(M_n,n)$, there exists a sequence $\left\{\{P(\cdot|m)\}_{m\in \{1,...,M_n\}},\; \{L_n(y^n)\}_{y^n\in\calY^n} \right\}$ of $(n,M_n,\epsilon_n,\zeta_n)$-codes such that $\lim_{n\rightarrow\infty}\epsilon_n=0$ and $\lim_{n\rightarrow\infty}\zeta_n=0$. 

Next we construct a sequence of $(n,M_n,\epsilon_n+\zeta_n,0)$-codes by replacing $L_n(y^n)$ with
\begin{flalign}
\tilde{L}_n(y^n)=\left\{\begin{array}{ll}
L_n(y^n) & \mbox{ if } |L_n(y^n)|\leq M_n/r(M_n,n)\\
\phi & \mbox { o.w. }
 \end{array}\right.,
\end{flalign}
where $\phi$ is the empty set. 

Therefore, as far as part (a) is concerned, we can assume without loss of generality that there exists a sequence of $(n,M_n,\epsilon_n,0)$-codes, i.e., with, $\zeta_n=0$ such that $\lim_{n\rightarrow\infty}\epsilon_n=0$. 

The rest of the proof refers to both part (a) and part (b) of the theorem. The proof relies on a modification of the Verd\'{u}-Han \cite{VerduHan1994} technique.

Let a sequence of $(n,M_n,\epsilon_n,\zeta_n)$-codes be given having vanishing error of the first kind, that is, 
$\lim_{n\rightarrow\infty }\epsilon_n=0$ and either (a) zero error of the second kind, i.e., $\zeta_n=0$ for all $n$, or (b) $\limsup_{n\rightarrow\infty } \frac{ \mathbb{E}\left(|L_n(Y^n)|\right)}{L(M_n,n)}=a <\infty$.

Let $(S,X^n,Y^n)$ be the random message, channel input sequence, and channel output sequence, respectively, and with a slight abuse of notation denote by $P(\cdot)$ their joint distribution. 
We have, 
\begin{flalign}
&\Pr\left(\frac{1}{n}\log \frac{P(S|Y^n)}{P(S)}\leq \frac{1}{n}\log\left( r(M_n,n)\right)-\gamma\right)\nonumber\\
&= \Pr\left(\frac{1}{n}\log \frac{P(S|Y^n)}{1/M_n}\leq \frac{1}{n}\log\left(\frac{M_n}{L(M_n,n)}\right)-\gamma\right)\label{eq: step 1}\\
&= \Pr\left(P(S|Y^n)\leq \frac{e^{-n\gamma}}{L(M_n,n)}\right)\nonumber\\
&\leq  \Pr(Y^n\in \calD_S^c)
+ \Pr\left(Y^n\in \calD_S, P(S|Y^n)\leq \frac{e^{-n\gamma}}{L(M_n,n)}\right)\nonumber\\
&=  \epsilon_n+ \sum_{m=1}^N \sum_{y^n\in \calD_m:\; P(m|y^n)\leq \frac{e^{-n\gamma}}{L(M_n,n)}} P_{Y^n}(y^n)P_{S|Y^n}(m|y^n)\nonumber\\
&\leq  \epsilon_n+ \frac{e^{-n\gamma}}{L(M_n,n)}\sum_{m=1}^N\mathbb{E}(\mathbbm{1}\{Y^n\in \calD_m\}) \nonumber\\
&=  \epsilon_n+ \frac{e^{-n\gamma}}{L(M_n,n)} \mathbb{E}\left(\sum_{m=1}^N\mathbbm{1}\{Y^n\in \calD_m\}\right) \nonumber\\
&= \epsilon_n+ \frac{e^{-n\gamma}}{L(M_n,n)} \mathbb{E}\left(|L_n(Y^n)|\right) \label{eq: long chain of inequalities}
\end{flalign}
where (\ref{eq: step 1}) follows since the messages are equiprobable and by definition of $r(M_n,n)$. 
Now, as far as part (a) is concerned, if $\zeta_n=0$ we have that $|L(Y^n)|\leq L(M_n,n)$ with probability one and therefore $\mathbb{E}\left(|L_n(Y^n)|\right) \leq L(M_n,n)$. Alternatively, for part (b) we have 
$\limsup_{n\rightarrow\infty } \frac{ \mathbb{E}\left(|L_n(Y^n)|\right)}{L(M_n,n)}=a <\infty$.
Thus, under either $\zeta_n=0$ or $\limsup_{n\rightarrow\infty } \frac{ \mathbb{E}\left(|L_n(Y^n)|\right)}{L(M_n,n)}<\infty$ we obtain
\begin{flalign}
\limsup_{n\rightarrow\infty} \frac{ \mathbb{E}\left(|L_n(Y^n)|\right) }{L(M_n,n)} \leq & \max\{1,a\}. 
\end{flalign}
As noted before, $S$ is also a function of the block length $n$. So, denote by $\bS$ the sequence of random variables $\{S_i\}_{i=1}^{\infty}$ where $S_n$ is the message that is transmitted over the channel $W^{(n)}$. 

Since by assumption the r.h.s.\ of (\ref{eq: long chain of inequalities}) vanishes; i.e., $\lim_{n\rightarrow\infty}\epsilon_n=0$, one must have that the l.h.s.\ of (\ref{eq: long chain of inequalities}) vanishes as well, and therefore by definition of the limit inferior in probability\footnote{
The limit inferior in probability \cite{VerduHan1994} of a sequence of random variables $X_n, n\geq 1$, denoted $p\mbox{-}\liminf X_n$, is the supremum of all $\alpha\in\mathbb{R}$ such that $\lim_{n\rightarrow\infty}\Pr\left\{X_n< \alpha\right\}=0$; i.e., 
$
p\mbox{-liminf } X_n=\sup\{\alpha:\; \limsup_{n\rightarrow\infty}\Pr\left\{X_n< \alpha\right\}=0\}
$.
\label{fn: liminf in prob}} we have
\begin{flalign}
\limsup_{n\rightarrow\infty} \frac{1}{n}\log\left( r(M_n,n)\right)-\gamma\leq& \sup_{P(\bX|\bS)} \underline{I}(\bS;\bY)\nonumber\\
\leq& \sup_{P(\bX)}  \underline{I}(\bX;\bY),\label{eq: agfdghg}
\end{flalign}
where $P(\bX|\bS)$ denotes a series of conditional distributions $P^{(n)}(X^n|S_n), n=1,2,...$, $P(\bX)$ denotes a series of distributions $P^{(n)}(X^n), n=1,2,...$, and 
where (\ref{eq: agfdghg}) follows since $S_n-X^n-Y^n$ is a Markov chain and from the data processing Theorem \cite[Theorem 9]{VerduHan1994}.

Since $\gamma$ can be made arbitrarily small, and since $ C=\sup_{P(\bX)}  \underline{I}(\bX;\bY)$ is the general formula for the channel capacity \cite{VerduHan1994} we have 
\begin{flalign}
&\limsup_{n\rightarrow\infty} \frac{1}{n}\log\left( r(M_n,n)\right)\leq  C.
\end{flalign}
\end{proof}

Now considering the special case $r(M_n,n)=M_n/e^{n\Theta}$, we can write the general formula for the capacity of the channel with the desired list size $|L(Y^n)|\leq e^{n\Theta}$, which we denote by $C( \Theta)$.  \begin{corollary}\label{th: c eq c plus theta}
For every channel, the supremum of achievable rates $R=\liminf_{n\rightarrow\infty} \frac{1}{n}\log M_n$ such that 
$\lim_{n\rightarrow\infty}\Pr(|L_n(Y^n)|>e^{n\Theta})=0$ and $\lim_{n\rightarrow\infty}\Pr(S\notin L_n(Y^n))=0$ is
\begin{flalign}\label{eq: exponential size of list equality}
C(\Theta)=   C+\Theta.
\end{flalign}
\end{corollary}

\section{Mismatched List Decoding}\label{sc: Mismatched List Decoding}
Thus far, we have considered list decoding with decoding regions $\{\calD_m\}$ that can be optimized to minimize the error probabilities. In this section we consider the case of decoding regions that are determined by a fixed decoding function, and we confine the discussion to deterministic codes. 
Nevertheless, we consider the setup of a general channel $\bW=\{W^{(n)}\}_{n=1}^{\infty}$ as in Section \ref{sc: An Extension to General Channels}.
For the sake of completeness, we describe the operation of the list decoder, as well as the classical mismatched decoder (list size equal to $1$).

A deterministic block-code $\calC_n$ of blocklength $n$ and $M_n$ messages, consists of $M_n$ $n$-vectors, that is, $\calC_n=\left\{x^n(m)\right\}_{m = 1}^{M_n}$, which represent $M_n$ different messages; i.e., it is defined by the encoding function
\begin{flalign}
f_n:\; \{1,...,M_n\} \rightarrow \calX^n.
\end{flalign}
As before, it is assumed that all possible messages are a-priori equiprobable; i.e., the random message $S_n$ satisfies $\Pr (S_n=m) = \frac{1}{M_n}$ for all 
$m\in\{1,...,M_n\}$.

Let a mapping
\begin{flalign}\label{eq: qn mapping}
q_n:\;  \calX^n\times \calY^n\rightarrow  \mathbb{R}
\end{flalign}
be given. 
A classical mismatched decoder outputs the message index $\hat{m}$ which maximizes the metric; that is,
\begin{flalign}\label{eq: mismatched decoder classical output}
\hat{m}=\mbox{argmax}_m q_n(x^n(m),y^n),
\end{flalign}
where if the maximizer is not unique, an error is declared. Let $\hat{S}(Y^n)$ stand for the output of the mismatched decoder. 

In the case of classical mismatched decoding (list size $1$), $M_n$ cannot grow faster than exponentially with $n$. 
We say that $R$ is an achievable rate with a decoding metric sequence $\bq=\{q_i\}_{i=1}^{\infty}$, if there exists a sequence of encoders $f_n:\{1,...,2^{nR}\}$ such that $\lim_{n\rightarrow\infty}\Pr\left(\hat{S}(Y^n)\neq S_n\right)=0$. 

The mismatch capacity $C_{\bq}(\bW)$ is the supremum of achievable rates using the decoder (\ref{eq: mismatched decoder classical output}). 

A mismatched list decoder is required to declare a list 
$L_n(y^n, q_n)\triangleq \{i_1,i_2,...,i_{|L(y^n,q_n)|}\}$ 
of messages which correspond to the highest metric values with the channel output, i.e., it should satisfy 
\begin{flalign}\label{eq: decoder decision rule}
&q_n(x^n(i),y^n)\geq q_n(x^n(j),y^n), \forall j\notin L_n(y^n,q_n),
\end{flalign}
where as before, the transmitted message is expected to belong to the list. 
Note that as in the matched case, we allow the list size to depend on $y^n$, and $M_n$ can potentially grow faster than exponentially with $n$. 
Note also that the requirement (\ref{eq: decoder decision rule}) does not define the decoder uniquely; the list size as a function of $y^n$ is not specified, and the decoder may be defined in several ways. For example, 
the decoder's output can be defined by the $M_n/r(M_n,n)$ messages having the highest metrics values, 
it can also be defined by a threshold $\tau(y^n)$ which the metric should cross for the message to be included in the list. The results of Sections \ref{sc: mim r l d} and
\ref{sc:  A General Formula for the Mismatch Capacity} concern any decoder that satisfies the requirement (\ref{eq: decoder decision rule}). Section \ref{sc: A Tight Expression for the Average Error Probability} analyzes the average probability of error of a decoder with a constant list size, i.e., a list size that is equal for all $y^n$.

We first present results that pertain to the counterpart of ratio list decoding to the mismatched decoding setup. Then we specialize the results to the case of a list size that grows exponentially with $n$.

\subsection{Mismatched Ratio List Decoding}\label{sc: mim r l d}

We next present counterparts to Definitions \ref{df: a code parameters}-\ref{df: maximal achievable rho} for the mismatched case. 
We begin with a formal definition of the average probability of error of the first and second kinds in mismatched ratio list decoding with a ratio function $r(M_n,n)$.
\begin{definition}\label{eq: P_e W calC_n q_n dfn}
For a given codebook $\calC_n$ of size $M_n=|\calC_n|$, let $\epsilon_{e,r}(W^{(n)},\calC_n,q_n)$ and $\zeta_{e,r}(W^{(n)},\calC_n,q_n)$ designate the average probability of error of the first and second kinds that  are incurred by the decoder $q_n$ employed on the output of the channel $W^{(n)}$ with list size $M_n/r(M_n,n)$, respectively, that is
\begin{flalign}
\epsilon_{e,r}(W^{(n)},\calC_n,q_n)
=&\Pr(X^n\notin L_n(Y^n,q_n))\nonumber\\
\zeta_{e,r}(W^{(n)},\calC_n,q_n)
=&\Pr(|L_n(Y^n,q_n)| > M_n/r(M_n,n)).
\end{flalign}
\end{definition}

We will also be interested in the probability of the event that the number of codewords whose metrics values are at least as high as the transmitted one exceeds the desired number $\frac{M_n}{r(M_n,n)}$, that is, 
\begin{flalign}\label{eq: P e r}
P_{e,r}\left(W^{(n)},\calC_n,q_n\right)\triangleq &\Pr\left\{|\{\bx'\in\calC_n :\; q_n(\bx',Y^n)\geq q_n(X^n,Y^n)\}| > \frac{M_n}{r(M_n,n)}\right\}.\end{flalign}
It holds that
\begin{flalign}\label{eq: P epsilon zeta}
P_{e,r}\left(W^{(n)},\calC_n,q_n\right)\leq \Pr(X^n\notin L_n(Y^n,q_n)\mbox{ or } |L_n(Y^n,q_n)| > M_n/r(M_n,n)).
\end{flalign}
In fact, $P_{e,r}\left(W^{(n)},\calC_n,q_n\right)$ is the probability of error in list decoding with constant list size in the sense of being equal for all $y^n$; that is, $L_n(y^n,q_n)\equiv M_n/r(M_n,n)$, $\forall y^n$.

\begin{definition}
A code $\calC_n$ is an $\left(n,M_n,\epsilon, \zeta ;q_n\right)$-code w.r.t.\ the ratio function $r(M_n,n)$, for  channel $W^{(n)}$ if 
it has $M_n$ codewords of length $n$, $\epsilon_{e,r}(W^{(n)},\calC_n,q_n)=
\epsilon$, and $\zeta_{e,r}(W^{(n)},\calC_n,q_n)=
\zeta$. 
\end{definition}

\begin{definition}\label{df: achievable rho jfbkdnj}
We say that a sequence $\{M_i\}_{i=1}^{\infty}$, $M_i\in\mathbb{N}$ is feasible for the channel w.r.t.\ a ratio sequence $r(M_n,n)$ with decoding metric sequence $\{q_i\}_{i=1}^{\infty}$ if there exists a sequence of $(n,M_n,\epsilon_n,\zeta_n;q_n)$-codes such that $\lim_{n\rightarrow\infty} \epsilon_n=0$ and $\lim_{n\rightarrow\infty} \zeta_n=0$.
\end{definition}

\begin{definition}
We say that $\rho$ is an achievable normalized log ratio for the channel with decoding metric sequence $\{q_i\}_{i=1}^{\infty}$ if there exists a ratio function $r(M_n,n)$ and a corresponding feasible sequence $\{M_i\}_{i=1}^{\infty}$, $M_i\in\mathbb{N}$ such that 
$\limsup_{n\rightarrow\infty}\frac{1}{n}\log r(M_n,n)= \rho$.
\end{definition}
\begin{definition}\label{df: maximal achievable rho irtghikshfksvh}
The {\it mismatch ratio-capacity} of the channel, $\rho_{\bq}(\bW)$, is the supremum of achievable normalized log ratios with a decoding metric sequence $\bq=\{q_i\}_{i=1}^{\infty}$ . 
\end{definition}

Let $\bq=\{q_i\}_{i=1}^{\infty}$, be a given sequence of decoding metrics, and consider the following random variable, which is a function of $(X^n,Y^n)$ 
\begin{flalign}
\Phi_{q_n}(X^n,Y^n)\triangleq 
 \mbox{Pr}\left\{q_n(\tilde{X}^n,Y^n)\geq q_n(X^n,Y^n)|X^n,Y^n\right\}
\end{flalign} where $X^n$ and $Y^n$ are the input and output channel vectors, respectively, and $\tilde{X}^n$ is independent of $(X^n,Y^n)$ and is distributed identically to $X^n$.

The following lemma extends \cite[Lemma 1]{SomekhBaruch_general_formula_IT2015}.
\begin{lemma}\label{lm: VerduHan Lemma}
Let $X^n$ be the random variable uniformly distributed over a code $\calC_n$ of blocklength $n$ and cardinality $M_n\triangleq |\calC_n|$, 
and $Y^n$ the output of a channel $W^{(n)}$ with $X^n$ as the input, then 
\begin{flalign}\label{eq: VerduHan UB yyy}
P_{e,r}(W^{(n)},\calC_n,q_n)
=&\Pr\left\{-\frac{1}{n} \log \left(\Phi_{q_n}(X^n,Y^n)\right) <\frac{1}{n}\log r(M_n,n) \right\}.
\end{flalign}
\end{lemma}
\begin{proof}
Note that
\begin{flalign}\label{eq: this equation explains strong converse}
\Phi_{q_n}(X^n,Y^n)=&\sum_{\bx'\in\calC_n:\; q_n(\bx',Y^n)\geq q_n(X^n,Y^n)} P^{(n)}(\bx') \nonumber\\
=&\frac{|\{\bx'\in\calC_n:\; q_n(\bx',Y^n)\geq q_n(X^n,Y^n)\}|}{M_n} \nonumber\\
\end{flalign}
where the last equality follows since $X^n$ is distributed uniformly over the codebook of size $M_n$. Hence, the right hand side of (\ref{eq: VerduHan UB yyy}) is equal to
\begin{flalign}
&\Pr\bigg\{-\frac{1}{n} \log \big(|\{\bx'\in\calC_n :\; 
q_n(\bx',Y^n)\geq q_n(X^n,Y^n)\}|\big) < -\frac{1}{n}\log \frac{M_n}{r(M_n,n)} \bigg\}\nonumber\\
&= \Pr\left\{|\{\bx'\in\calC_n :\; q_n(\bx',Y^n)\geq q_n(X^n,Y^n)\}| > \frac{M_n}{r(M_n,n)}\right\}\nonumber\\
&= P_{e,r}\left(W^{(n)},\calC_n,q_n\right).
\end{flalign}
This concludes the proof of Lemma \ref{lm: VerduHan Lemma}. 
\end{proof} 

Based on Lemma \ref{lm: VerduHan Lemma} we can now present a ratio list decoding theorem, which is the counterpart of Theorem \ref{th: main theorem} for the mismatched case. 

Recall the definition of $\rho_r$ (\ref{eq: rho r dfn}). 
\begin{theorem}\label{th: main theorem mis}
Let $\{M_i\}_{i=1}^{\infty}$, $M_i\in\mathbb{N}_i$, be a given sequence and let $r(M_n,n)$ be a ratio function. 

(i) Converse part: If the sequence $\{M_i\}_{i=1}^{\infty}$ is feasible for channel $\bW$ w.r.t.\ the
 ratio function $r(M_n,n)$ and a decoding metric sequence $\bq$, then 
\begin{flalign}\label{eq: converse part ofTh1 mis}
\rho_r\leq  C_{\bq}(\bW).
\end{flalign}

(ii) Direct part: If the sequence $\{M_i\}_{i=1}^{\infty}$ satisfies 
\begin{flalign}\label{eq: direct part ofTh1  b}
0<\rho_r<C_{\bq}(\bW),
\end{flalign}
then it is feasible for the channel $\bW$ w.r.t.\ ratio function $r(M_n,n)$ and metric sequence $\bq$. 
\end{theorem}
Recall that $C_{\bq}(\bW)$ denotes the mismatch capacity of the channel $\bW$ with decoding metric sequence $\bq$. 
The theorem straightforwardly implies the following corollary.
\begin{corollary}\label{th: rho theorem mismatch}
For every channel $\bW$
\begin{flalign}
 \rho_{\bq}(\bW)=C_{\bq}(\bW). 
 \end{flalign}
\end{corollary}
We next prove Theorem \ref{th: main theorem mis}.
The converse part follows from (\ref{eq: converse part ofTh1 mis}) and the direct part follows by using the ratio function $r(M_n,n)=M_n$ (ordinary decoding) and by definition of $C_{\bq}(\bW)$.
\begin{proof}
We begin with the converse part. 
Assume in negation that $\{M_i\}_{i=1}^{\infty}$ is a feasible sequence w.r.t.\ $r(M_n,n)$ with a decoding metric $\bq$ such that for some sufficiently small $\gamma>0$ and for infinitely many $n$'s $\frac{1}{n}\log(r(M_n,n))
\geq C_{\bq}(\bW)+\gamma$. 
By the feasibility assumption and by (\ref{eq: P epsilon zeta}), there exists a sequence of codes $\calC_1,\calC_2,... $, with $|\calC_n|=M_n$ satisfying $\limsup_{n\rightarrow\infty} P_{e,r}(W^{(n)},\calC_n,q_n)=0$, 
Thus, from Lemma 
\ref{lm: VerduHan Lemma} we have for infinitely many $n$'s,
\begin{flalign}
&P_{e,r}(W^{(n)},\calC_n,q_n )
\geq 
\Pr\left\{-\frac{1}{n} \log\Phi_{q_n}(X^n,Y^n) < C_{\bq}(\bW)+\gamma  \right\}.
\label{eq: VerduHan UB second appearance}
\end{flalign}
Since the general formula for the mismatch capacity (\hspace{1sp}\cite{SomekhBaruch_general_formula_IT2015}) is given by 
\begin{flalign}
C_{\bq}(\bW)=\sup_{\bP} \mbox{p-liminf} -\frac{1}{n} \log\Phi_{q_n}(X^n,Y^n),
\end{flalign} 
(see Footnote \ref{fn: liminf in prob} for the definition of $\mbox{p-liminf}X_n$.) the r.h.s.\ of (\ref{eq: VerduHan UB second appearance}) is bounded away from zero for infinitely many $n$'s, and hence $P_{e,r}(W^{(n)},\calC_n,q_n)$ cannot vanish (and by (\ref{eq: P epsilon zeta}), neither can $\epsilon_{e,r}(W^{(n)},\calC_n,q_n)$ and $\zeta_{e,r}(W^{(n)},\calC_n,q_n)$ both be vanishing sequences) in contradiction to the assumption on the feasibility of $\{M_i\}_{i=1}^{\infty}$.

The direct part of Theorem \ref{th: main theorem mis} follows similarly to the proof of Theorem \ref{th: main theorem} (i), with the exception that now $C$ should be replaced by the mismatch capacity $C_{\bq}(\bW)$ and $\hat{S}(Y^n)$ is the output of the mismatched decoder. 

\end{proof}

\subsection{A General Formula for the Mismatch Capacity with List Decoding}\label{sc:  A General Formula for the Mismatch Capacity}

For the special case of a list size that grows no faster that exponentially with $n$; that is $L(M_n,n)= e^{n\Theta_n}$, where $\Theta_n\in[0,C]$, we define the mismatch capacity with list size exponent $\Theta_n$. 
\begin{definition}\label{df: definition 7}
A rate $R>0$ is an achievable rate for the channel $\bW$ with a decoding metric sequence $\bq$ and a list size sequence $e^{n\Theta_n}$ (where $\Theta_n\in[0,C]$)  if there exists a sequence of codes $\{\calC_n\}_{n=1}^{\infty}$ such that $\calC_n$ is an $(n,e^{nR},\epsilon_n,\zeta_n;q_n)$-code w.r.t.\ the ratio function $r(M_n,n)=M_n\cdot e^{-n\Theta}$ for the channel $W^{(n)}$ and it holds that $\lim_{n\rightarrow\infty}\epsilon_n=0$ and $\lim_{n\rightarrow\infty}\zeta_n=0$. 
\end{definition}
\begin{definition}
The capacity of the channel $\bW=\{W^{(n)}\}_{n=1}^{\infty}$ with a decoding metric sequence $\bq=q_1,q_2,... $ and a list size exponent sequence $\bTheta=\Theta_1,\Theta_2,... $ denoted $C_{\bq}(\bW,\bTheta)$, is the supremum of the achievable rates as in Definition \ref{df: definition 7}. \end{definition}

The multi-letter expression for the mismatch capacity with the list size exponent sequence $\bTheta$ is stated in the following theorem.
\begin{theorem}\label{th: General formula expression}
The mismatch $\bq$-capacity of the channel $\bW$ with a list-size exponent sequence $\bTheta$ is given by \begin{flalign}\label{eq: the General formula ksdjhfkjhsk}
C_{\bq}(\bW,\bTheta)= & \sup_{\bP} \mbox{p-}\liminf  -\frac{1}{n}\log\left(\Phi_{q_n}(X^n,Y^n)\right)+ \Theta_n
\end{flalign}
where the supremum can be restricted to sequences of distributions that are uniform over their supports. 
\end{theorem}

This is a special case of Corollary \ref{th: rho theorem mismatch}, and thus follows straightforwardly from it\footnote{See also the proof of this special case in \cite[Theorem 1]{SomekhBaruch_list_ISIT2015}.}; simply   
substitute $\rho_{\bq}(\bW)$ by $C_{\bq}(\bW,\bTheta)-\liminf_{n\rightarrow\infty}\Theta_n$ and note that 

\noindent $C_{\bq}(\bW)= \sup_{\bP} p\mbox{-}\liminf  -\frac{1}{n}\log\left(\Phi_{q_n}(X^n,Y^n)\right)$. 

The following corollary is therefore the counterpart of Corollary \ref{th: c eq c plus theta} to the mismatched decoding setup. \begin{corollary}\label{cr: coroallydfkjv}
If $\forall n, \Theta_n=\Theta$, one has
\begin{flalign}\label{eq: the General formula ksdjhfkjhsk kejahsfjk}
C_{\bq}(\bW,\Theta)=  C_{\bq}(\bW)  + \Theta.
\end{flalign}
where $C_{\bq}(\bW)$ is the mismatch capacity, and when $\liminf_{n\rightarrow \infty}\Theta_n=0$, one has
\begin{flalign}\label{eq: the General formula ksdjhfkjhsk}
C_{\bq}(\bW,\bTheta)=  C_{\bq}(\bW).
\end{flalign}
\end{corollary}
Corollary \ref{cr: coroallydfkjv} implies that for every channel and every metric sequence, if the list size grows sub-exponentially, the list decoding capacity is equal to the ordinary capacity \cite{SomekhBaruch_general_formula_IT2015}.

\subsection{Properties of the Average Error Probability in Mismatched List Decoding with Constant List Size}\label{sc: A Tight Expression for the Average Error Probability}

Thus far, we have studied achievable rates for list decoding. In this section we present results that concern the average probability of error in list decoding with or without mismatch. 

The exact error probability depends on the actual decision rule that determines the list $L(y^n,q_n)$, and in particular whether or not $L(y^n,q_n)$ is constant for all $y^n$ or determined by a threshold level on the metric value. Thus, in this section we confine attention to the case of a constant list size that is
\begin{flalign}
L(y^n,q_n)=e^{n\Theta_n}, \; \forall y^n\in\calY^n
\end{flalign}
in which case the average probability of error is given by 
\begin{flalign}
P_e^{\Theta_n}(W^{(n)},\calC_n,q_n)\triangleq \left. P_{e,r}(W^{(n)},\calC_n,q_n)\right|_{r(M_n,n)=M_ne^{-n\Theta_n}},
\end{flalign}
where $P_{e,r}(W^{(n)},\calC_n,q_n)$ is defined in (\ref{eq: P e r}). 

\subsubsection{A Tight Expression for the Average Error Probability with Mismatched List Decoding}
Let $\calG_n(R)$ be the set of codebooks of block length $n$ and rate $R$. Denote the infimum of the achievable average error probability with equiprobable codewords at rate $R$ and block length $n$ by
\begin{flalign}
\calE_{q_n}^{(n)}(R)=&\inf_{\calC_n\in\calG_n(R)}\epsilon(W^{(n)},\calC_n,q_n),
\end{flalign}
where $\epsilon(W^{(n)},\calC_n,q_n)$ denotes the average probability of error for ordinary (list size equals one) mismatched decoding.

 Further, denote the equivalent quantity for list decoding of list-size $e^{n\Theta_n}$ \begin{flalign}
\calE_{q_n}^{(n)}(R,\Theta_n)=&\inf_{\calC_n\in\calG_n(R)}P_e^{\Theta_n}(W^{(n)},\calC_n,q_n).
\end{flalign}

Let $\calP_n(R)$ be the set of distributions which are uniform over a subset of $\calX^n$ whose size is $e^{nR}$. 
For the sake of convenience we use the abbreviation for $P^{(n)}$ the distribution of $X^n$ 
\begin{flalign}
\Phi_{q_n}(P^{(n)})\triangleq \Phi_{q_n}(X^n,Y^n). 
\end{flalign}
\noindent From Lemma \ref{lm: VerduHan Lemma} we obtain the following straightforward result.
\begin{theorem}
For all $R,\Theta_n$
\begin{flalign}\label{eq: VerduHan UB}
&\calE_{q_n}^{(n)}(R,\Theta_n)
=
\inf_{P^{(n)}\in\calP_n(R)}\Pr\left\{-\frac{1}{n} \log \left(\Phi_{q_n}(P^{(n)})\right) < R-\Theta_n \right\}.
\end{flalign}
\end{theorem}
As a special case, the following theorem follows (it is also a straightforward result of \cite[Lemma 1]{SomekhBaruch_general_formula_IT2015}). 
\begin{theorem}
For all $R$
\begin{flalign}\label{eq: VerduHan UB}
\calE_{q_n}^{(n)}(R)=
\inf_{P^{(n)}\in\calP_n(R)}\Pr\left\{-\frac{1}{n} \log \left(\Phi_{q_n}(P^{(n)})\right) < R \right\}.
\end{flalign}
\end{theorem}
The following inequality holds:
\begin{lemma}
For all $R,n$, 
\begin{flalign}
\calE_{q_n}^{(n)}(R+\Theta_n,\Theta_n)\leq  \calE_{q_n}^{(n)}(R).
\end{flalign}
\end{lemma}
Similar to the proof of the direct part of Theorem \ref{th: General formula expression}, the lemma follows by the simple observation that given a codebook $\calC_n$ of size $e^{nR}$ one can create a codebook $\calC_n'$ of size $e^{ n(R+\Theta_n)}$ containing $e^{nR}$ sets of identical $e^{n\Theta_n}$ codewords, and that $\epsilon(W^{(n)},\calC_n,q_n)=P_e^{\Theta_n}(W^{(n)},\calC_n,q_n)$. 

\subsubsection{A Random Coding Lower Bound on the Average Probability of Error with Mismatched List Decoding}\label{sc: A Random Coding Lower Bound}
The next result is an upper bound on the average error probability obtained by random coding. 
\begin{theorem}\label{th: random coding list bound}
For all $R,\Theta_n$
\begin{flalign}
&\calE_{q_n}^{(n)}(R,\Theta_n)\leq   \inf_{P^{(n)}\in\calP(\calX^n)} \bigg[\mathbb{E}\Big[ e^{-e^{nR}\cdot D\left(e^{-n(R-\Theta_n)}\| \Phi_{q_n}(P^{(n)}) \right)} \nonumber\\
& \hspace{4cm}
\times 1\{\Phi_{q_n}(P^{(n)}) < e^{-n(R-\Theta_n)}\}\Big]
\mathbb{E}\left[1\{\Phi_{q_n}(P^{(n)})\geq e^{-n(R-\Theta_n)}\}\right]\bigg].
\end{flalign}
\end{theorem}
\begin{proof}
Consider the chain of inequalities that was derived in \cite[Equations (8)-(18)]{Merhav_List_Decoding_IT_2014}, which can be phrased as follows: 
for all positive integers $M\geq L$ and $\Phi\in[0,1]$, 
\begin{flalign}
&\sum_{k=L}^M {M \choose k} \cdot \Phi^k\cdot\left[1-\Phi \right]^{M-k} \nonumber\\
&\leq 1\left\{\Phi \geq L/(M-1)\right\}
+  \exp\left\{-MD\left(L/M\| \Phi \right)\right\}\cdot 1\left\{\Phi<L/(M-1)\right\}\nonumber\\
&\leq  \exp\left\{-L\left[\ln (L)-\ln(M\Phi) -1\right]_+\right\},
\end{flalign}
where $D(p\|q)$ is the binary divergence 
and $|t|_+=\max\{0,t\}$. 
Denote the random variable $Z_i= 1\{q_n(\tilde{X}^n(i),Y^n)\geq q_n(X^n,Y^n)\}$ where $\tilde{X}^n(i)$ is the random $i$-th codeword and $X^n$ is the transmitted one. 
Similar to the derivation in \cite{Merhav_List_Decoding_IT_2014} we obtain that the random coding average probability of error $\calE_{RC, q_n}^{(n)}(R,\Theta_n)$ achieved when the codewords are drawn i.i.d.\ $P^{(n)}$ is upper bounded as follows:
\begin{flalign}
&\calE_{RC, q_n}^{(n)}(R,\Theta_n)\nonumber\\
&= 
\mathbb{E}\left\{\Pr\left\{\left.\sum_{i=2}^{e^{nR}}Z_i \geq e^{n\Theta_n}\right|X^n,Y^n\right\}\right\}\nonumber\\
&= \mathbb{E}\bigg[\sum_{k=e^{n\Theta_n}}^{M_n} {M_n \choose k} \left[\Phi_{q_n}(P^{(n)})\right]^k\cdot \left[1-\Phi_{q_n}(P^{(n)}) \right]^{M_n-k}\bigg] \nonumber\\
&\leq   \mathbb{E}\bigg[ e^{-M_nD\left(e^{-n(R-\Theta_n)}\| \Phi_{q_n}(P^{(n)}) \right)} 
\times1\{\Phi_{q_n}(P^{(n)})< e^{-n(R-\Theta_n)}\}\bigg]
+  \mathbb{E}\left[1\{\Phi_{q_n}(P^{(n)})\geq e^{-n(R-\Theta_n)}\}\right].
\end{flalign}
\end{proof}
In the matched DMC case, this bound was shown to prove the tightness of the Shannon-Gallager-Berlekamp bound \cite{ShannonGallagerBerlekamp1967}
across the relevant
range of rates, $(\Theta,C + \Theta)$, where $C$ is the channel capacity. It would be interesting to see whether this result has an equivalent in the mismatched setup.

\subsubsection{A Lower Bound on the Average Probability of Error in Mismatched List Decoding for Rates Above Capacity}\label{sc: An Equivalent of Fano}
As a result of Fano's Inequality, in the matched DMC case it holds that (see, e.g., \cite[Eq. (7.103)]{CoverThomas2006})
\begin{flalign}
P_e^{(n)}\geq 1-\frac{C}{R}-\frac{1}{nR},
\end{flalign} 
where $C$ is the channel capacity and $P_e^{(n)}$ is the average probability of error obtained by a code of rate $R$ and a maximum likelihood decoder. 
Consequently, 
\begin{flalign}\label{eq: liming Fano}
\liminf_{n\rightarrow \infty}P_e^{(n)}\geq 1-\frac{C}{R}.
\end{flalign}
We present the following related result. For the simplicity of the presentation, in this section we consider a list size exponent $\Theta$ that does not depend on $n$.  
\begin{theorem}\label{th: Fano Theorem}
For every channel $W^{(n)}$, codebook $\calC_n$ and a metric $q_n$
\begin{flalign}\label{eq: skjbvkjbh}
&P_e^{\Theta}(W^{(n)},\calC_n,q_n)\geq 1-\frac{1}{R-\Theta}\mathbb{E}\left\{-\frac{1}{n}\log \Phi_{q_n}(P^{(n)}) \right\},
\end{flalign}
where $P^{(n)}$ is uniform over $\calC_n$. 
\end{theorem}
Before we present the proof of Theorem \ref{th: Fano Theorem} some comments are in order. 
Theorem \ref{th: Fano Theorem} can be regarded as an extension of (\ref{eq: liming Fano}) to the mismatched case with list decoding. To see this, note that by substituting $\Theta=0$  in (\ref{eq: skjbvkjbh}) we obtain
\begin{flalign}\label{eq: jksvjdfn}
&\liminf_{n\rightarrow \infty}{\cal E}_{q_n}^{(n)}(R)\geq 1-\limsup_{n\rightarrow \infty}\sup_{P^{(n)}} \frac{1}{R}\mathbb{E}\left\{-\frac{1}{n}\log \Phi_{q_n}(P^{(n)}) \right\},\nonumber
\end{flalign}
and $\limsup_{n\rightarrow \infty}\sup_{P^{(n)}} \mathbb{E}\left\{-\frac{1}{n}\log \Phi_{q_n}(P^{(n)}) \right\}$ coincides with the mismatch capacity if the channel satisfies the strong converse property (see  \cite{SomekhBaruch_general_formula_IT2015}), e.g., a DMC with a matched metric.

Note the following corollary
\begin{corollary}\label{cr: trivial corollary}
For every channel $\bW$ and metrics sequence $\bq$,
\begin{flalign}
C_{\bq}(\bW)\leq &\liminf_{n\rightarrow\infty}\sup_{P^{(n)}\in\calP(\calX^n)}\mathbb{E}\left\{-\frac{1}{n}\log \Phi_{q_n}(P^{(n)}) \right\},\nonumber
\end{flalign}
where the supremum can be taken over distributions which are uniform over a subset of $\calX^n$.
\end{corollary}
Corollary \ref{cr: trivial corollary} was derived in \cite{SomekhBaruch_general_formula_IT2015} using a different line of proof. 
As a special case, consider the DMC $W$ with an erasures-only metric (without list decoding); i.e., $q_{eo}(x^n,y^n)=1\{W(y^n|x^n)>0\}$. It was proved in \cite[Theorem 3]{BunteLapidothSamorodintskyIT2014} that the erasures-only capacity, $C_{q_{eo}}(W)$ satisfies
 $C_{q_{eo}}(W)
=  \lim_{n\rightarrow\infty}  \max_{P^{(n)}} -\frac{1}{n} \mathbb{E}\left(
\log\Phi_{q_{eo}}(X^n,Y^n) \right)$.
 Combining this with Theorem \ref{th: Fano Theorem} we obtain the following corollary which yields, as a special case, a lower bound on the average probability of error at rates above capacity. 
\begin{corollary}\label{cr: trivial corollary eo}
The erasures-only capacity of the DMC satisfies
\begin{flalign}
\calE_{eo}(R,\Theta) \geq 1-\frac{C_{q_{eo}}(W)}{R-\Theta},
\end{flalign}
where $\calE_{eo}(R,\Theta)=\liminf_{n\rightarrow \infty}  \calE_{q_{eo}}^{(n)}(R,\Theta)$.
\end{corollary}
We continue with the proof of Theorem \ref{th: Fano Theorem}. 
\begin{proof}
Let a codebook $\calC_n=\left\{\bx_m\right\}_{m=1}^{M_n}$ of size $M_n=e^{nR}$ be given and let $X^n$ be distributed uniformly over $\calC_n$. 
We get that
\begin{flalign}
&1-P_e^{\Theta}(W^{(n)},\calC_n,q_n)\nonumber\\
&=\mathbb{E}\left(\mbox{Pr}\left\{\left.  \sum_{\bx':q_n(\bx',Y^n)\geq q_n(X^n,Y^n)}P^{(n)}(\bx') \leq\frac{e^{n\Theta}}{M_n}\right |X^n\right\}\right)\nonumber\\
&=\mathbb{E}\left(\mbox{Pr} \left\{\left.
\frac{1}{n}\log \Phi_{q_n}(q_n(X^n,Y^n),P^{(n)},Y^n) \leq \Theta-R\right |X^n\right\}\right)\nonumber\\
&\leq\frac{1}{M_n}\sum_{m=1}^{M_n}\sum_{\by}W^{(n)}(\by|\bx_m)\frac{-\frac{1}{n}\log \Phi_{q_n}(q_n(\bx_m,\by),P^{(n)},\by)}{R-\Theta}\nonumber\\
&= \frac{1}{R-\Theta}\mathbb{E}\left\{-\frac{1}{n}\log \Phi_{q_n}(q_n(X^n,Y^n),P^{(n)},Y^n) \right\},
\end{flalign}
and this concludes the proof of Theorem \ref{th: Fano Theorem}.  
\end{proof}

\section{Summary and Discussion}\label{sc: discussion}
In this paper, the notion of ratio list decoding was introduced as a generalization of list decoding to the case of a list size that is specified as a function of both the number of messages and the block length. 
For certain choices of these functions, the number of messages in reliable ratio list decoding can grow faster than exponentially with the block length. 
For example, for $r(M_n,n)=\log(M_n)$, the number of reliably transmitted codewords with list decoding can grow up to double-exponentially with $n$; that is, the supremum of the normalized iterated logarithm of $r(M_n,n)$, $\frac{1}{n}\log \log M_n$, is equal to the Shannon capacity. 
We treated the general $r(M_n,n)$ case, including cases such as $\log(M_n), N^{\alpha}, \alpha\in[0,1]$, etc. 
This is particularly relevant for applications that can tolerate a small ratio of codebook size to list size.

Furthermore, we distinguished between feasible and non-feasible sequences $\{M_i\}_{i=1}^{\infty}$ w.r.t.\ a rate function $r(M_n,n)$, where $M_n$ is the number of messages as a function of the block length.  We showed that for every channel, if $\rho_r\triangleq \limsup_{n\rightarrow\infty}\frac{1}{n}\log r(M_n,n)>C$ reliable list decoding cannot occur, and if $\rho_r\in(0,C)$, there exists a sequence of codes having $M_n$ messages for block length $n$ which enables reliable list decoding w.r.t.\ the ratio function $r(M_n,n)$. 

The supremum of achievable normalized codebook to list size log ratios $\rho_{sup}(\bW)$ was characterized as follows:
(a) in the case of a general channel with stochastic or deterministic encoding - it was shown to be equal to the Shannon channel capacity; 
(b) in the case of a general channel with deterministic encoding and a list which is determined by a fixed metric, it was shown to be equal to the mismatch capacity. 
In either case, the quantity $\rho_{sup}(\bW)$ is therefore equal to the supremum of the bits of information per channel use that can be transmitted reliably over the channel $\bW$.

\section{Acknowlegdements}
The author would like to thank Igal Sason and the anonymous reviewers for valuable comments which improved the quality of the paper.


\end{document}